\documentclass[runningheads]{llncs}
\usepackage{graphicx}
\usepackage{wrapfig}
\usepackage{float}
\usepackage{amsmath}
\usepackage{amssymb}
\usepackage{mathtools}
\usepackage[english]{babel}
\usepackage{cite}
\usepackage{textcmds}
\usepackage[utf8]{inputenc}
\usepackage{lineno}

\DeclareMathOperator{\compose}{{{\normalfont\text{\textbullet}}}}
\newcommand{\defeq}{\mathrel{=_{\text{def}}}}

\begin{document}
\title{Once and for all: how to compose modules --\newline The composition calculus}
\titlerunning{The composition calculus}

\author{Peter Fettke\inst{1,2}\orcidID{0000-0002-0624-4431} \and
Wolfgang Reisig\inst{3}\orcidID{0000-0002-7026-2810}}
\authorrunning{P. Fettke, W. Reisig}

\institute{German Research Center for Artificial Intelligence (DFKI), Saarbr\"ucken, Germany \\
\email{peter.fettke@dfki.de}\\ \and
Saarland University, Saarbr\"ucken, Germany \\ \and
Humboldt-Universität zu Berlin, Berlin, Germany \\ 
\email{reisig@informatik.hu-berlin.de}}

\maketitle % typeset the header of the contribution
\begin{abstract}
Computability theory is traditionally conceived as the \emph{the} theoretical basis of informatics. Nevertheless, numerous proposals transcend computability theory, in particular by emphasizing interaction of modules, or components, parts, constituents, as a fundamental computing feature. In a technical framework, interaction requires \emph{composition} of modules. Hence, a most abstract, comprehensive theory of modules and their composition is required. To this end, we suggest a minimal set of postulates to characterize systems in the digital world that consist of interacting modules. For such systems, we suggest a calculus with a simple, yet most general composition operator which exhibits important properties, in particular associativity. We claim that this composition calculus provides not just another conceptual, formal framework, but that essentially all settings of modules and their composition can be based on this calculus. This claim is supported by a rich body of theorems, properties, special classes of modules, and case studies.

\end{abstract}

\section{Introduction}\label{sec:1}

\subsection{The problem}\label{sec:1:1}

Scientific areas come frequently with specific foundational theories. For example, classical engineering areas employ \emph{the} calculus, covering the real numbers, continuous functions, differential equations, et cetera. For informatics, right from the start, the theory of \emph{computability theory} has been accepted as its theoretical basis. This choice has been motivated and justified by the observation that the capability of a software program could plausibly be abstracted as a computable function over the words of an alphabet, or over the natural numbers. Several attempts to declare more functions as effectively algorithmically computable, in an intuitive sense, were doomed to fail.

Meanwhile, informatics penetrates many aspects of everyday life. Computers became components of a new, digital world. This world includes systems consisting of decentralized modules, including mechanical machinery, organizations, and persons in distinguished roles. The theory of computable functions no longer suffices as a basis to understand this world. This rises the quest for a new theoretical foundation of informatics. So, a fundamentally new question arises: what aspects are common to all systems in the digital world, and how can they be abstracted into a framework that provides a formal basis for each system in this world?

Literature shows a general understanding that a system is composed of sub-systems, sometimes called \emph{parts}, \emph{components}, \emph{constituents}, \emph{modules}, or similarly. Such modules may be extremely heterogeneous; their interior structures or behavior share merely no common aspect at all. Nevertheless, each module may belong to a bigger system, i.e., be composed with other modules. Hence, \emph{composition} of modules is a fundamental and universal operation on systems in the digital world. 

So far, there is no unified understanding, let alone a framework or a theory, for all the various specific composition and refinement operations, as available in literature. A naïve approach would define a module as a graph structure with an interface, where the interface consists of labeled nodes. Composition of two such modules yields then a module again, gained by merging correspondingly labeled interface elements of the modules. This starting point is promising; but it is fundamentally flawed, as we will see later. This flaw may have stipulated the presumption that it is not promising to strive at a general, unifying theory of composition.
 
Against this background, here we present the \emph{composition calculus}, a framework of modules and their composition. It is \emph{universal} in the sense that we claim that each reasonable concept of modules and their composition can be abstracted in this setting. This claim is supported by case studies in various areas. In this contribution we show that it is also supported by a deep theoretical basis, a rich body of theorems, properties, and special cases.

\subsection{The idea of the composition calculus}\label{sec:1:2}

Merely all modeling techniques organize composition of modules according to the same principle: Each module has labeled interface elements. Composition of two modules $M$ and $N$ merges interface elements of $M$ and $N$ with “corresponding” or “complementary” labels. A typical and most prominent example are process algebras, in particular Robin Milner's calculus of communicating systems (CCS).

The composition calculus follows the dichotomy of complementary labeled interface elements in a slightly revised form but with decisively different consequences: there are two kinds of labeled interface elements, “left” and “right” elements. Composition $M \compose N$ of two modules $M$ and $N$ is then defined by merging each right interface element of $M$ with an equally labeled left interface of $N$. This reflects the observation that in the everyday world, many systems have modules with two kinds of interface elements. Typical examples include: an electric extension flex has a socket and a plug; a magnet has a positive and a negative pole; a Lego brick has nubs and cavities; clothing has buttons and buttonholes; a stock market has sellers and buyers; a factory has providers and requestors. Modules of the same kind can be composed along their complementary elements. Composition then represents activities such as to \emph{link electric wires}, to \emph{click magnets together}, to \emph{stick Lego bricks together}, to \emph{fasten clothes}, to \emph{trade shares}, to \emph{transfer items}, to \emph{operate a business}, et cetera.

Upon constructing the composition $M \compose N$, for a right interface element of $M$ there may exist many or no corresponding left interface elements of $N$. These cases are handled in such a way that composition is well defined, i.e. results in a unique module $M \compose N$, and that composition is associative, i.e one can safely write $M_1 \compose M_2 \compose \dots \compose M_n$, without any brackets.

Obviously, a module consists of elements in its interior, and elements in its interface. The most important operation on modules is their composition, resulting again in a module. In this paper, we study properties of module composition and of classes of modules, independent from specific properties or structures of the interior of modules. The only aspect that matters are properties of interfaces. As a technical convenience, the interior of modules is assumed to be graph-like. But this is not essential. Texts, programs et cetera would likewise do. They may be connected to interface elements in any way

This version of composition has been described and applied in many contributions, including \cite{HERAKLIT_website} and \cite{Fettke_Reisig_24} among others. The composition calculus evolves a rich body of properties and special cases, useful in applications. This suggests that the composition calculus is not just one choice, but a fundamental concept, indeed.

\subsection{This contribution}\label{sec:1:3}

This paper concentrates on properties of interfaces and classes of interfaces that are independent of the interior of modules. This is sensible because the definition of composition of modules depends only on the modules’ interfaces.

In Section 2, the notions of interfaces and modules are introduced, as well as the composition of modules. The interface of a composed module $M \compose N$ consists of interface elements of $M$ and of $N$. Its structure only depends on the structure of the interfaces of $M$ and of $N$. Furthermore, the interior of a composed module $M \compose N$ is more or less the union of the interior of $M$ and the interior of $N$, together with merged interface elements of $M$ and $N$.

Section 3 presents a number of properties of module composition in an algebraic framework. As mentioned, we consider properties that depend only on the structure of the interfaces, independent of the interior of the involved modules. Section 4 then studies particular properties of classes of modules. Section 5 relates the composition calculus to modules and composition operations that can be found in literature.

\section{The notion of modules}\label{sec:2}

\subsection{Motivation}\label{sec:2:1}
We strive at a most liberal notion of \emph{modules} and their \emph{composition}. This should include each conceivable specific conception of module and composition, as a special case. To be more specific, we formulate four \emph{postulates}, stating properties that each kind of modules meet:

\begin{enumerate}
\item A module includes any set of “elements” and any kind of relations among those elements. Some elements of a module $M$ belong to the \emph{interface} of $M$, the rest to the \emph{interior} of $M$.

\item Each interface element caries a kind of \emph{label}. Two modules are composed according to a rule about labeled elements of their interfaces. 

\item Composition should be technically simple, but still be capable of expressing any aspect considered relevant when modules are composed. Any two modules $M$ and $N$ can be composed, at least technically; though the composed module $M \compose N$ may not always be useful.

\item The composition of modules $M_1, \dots, M_n$ may be written $M_1 \compose M_2 \compose \dots \compose M_n$, i.e. brackets must not influence composition.
\end{enumerate}

Literature on modules, components, interface languages, et cetera and their composition accompanies composition frequently by additional requirements, depending on the actual states or properties of the involved modules. Composition may include alternative or unspecified behavior, et cetera. All this must be adequately captured by a fitting calculus of composition. 

There are not many frameworks that would meet the above postulates to a satisfactory degree. In this contribution, a framework is systematically built up that meets the above postulates.

\subsection{Interfaces and matches}\label{sec:2:2}
We start with the notion of an \emph{interface}. An interface of a module is a set of labeled and ordered module elements called \emph{gates}. Modules are composed along their interfaces by merging “matching” gates. 

\begin{definition}[interface] Let $\Lambda$ be a set with elements called “labels”. An \emph{interface over $\Lambda$} is a totally ordered finite set, its \emph{carrier set}, where each element, called a \emph{gate}, caries a label from $\Lambda$. The order is denoted as usual by “$<$”.
\end{definition}

In the rest of this contribution, we frequently assume a set $\Lambda$ of labels without further notice.

In graphical representations, the gates of an interface are usually vertically ordered, with the first, smallest element on top. Figure~\ref{fig:01} shows two interfaces, $A$ and $B$, over the label set $\Lambda = \{\alpha, \beta\}$.

Interfaces will be parts of modules and will be used to compose modules. To represent composition, gates of two interfaces are merged in case they are \emph{matching}. Gates of two interfaces \emph{match}, if they coincide in labeling \emph{and} order:

\begin{definition}[match]
Let $A$ and $B$ be two interfaces over $\Lambda$, let $a \in A$ and $b \in B$

\begin{enumerate}

\item The set $\{a, b\}$ is a \emph{match of $A$ and $B$}, if for some label $\lambda$, both $a$ and $b$ are $\lambda$-labeled, and the number of $\lambda$-labeled gates that are smaller than $a$ in $A$ is equal to the number of $\lambda$-labeled gates that are smaller than $b$ in $B$.

\item Let \emph{matches$(A, B)$} be the set of all matches of $A$ and $B$. 

\item A gate of $A$ that does not belong to a match of $A$ and $B$ is \emph{match free with respect to $B$}. Symmetrically, a gate of $B$ that does not belong to a match of $A$ and $B$ is match free with respect to $A$.

\item Let \emph{matchfree$(A, B)$} be the set of all gates of $A$ that do not belong to $a$ match with $B$.

\end{enumerate}
\end{definition}

For example, in Figure~\ref{fig:01}, the interfaces $A$ and $B$ have two matches: $\{a, d\}$ and $\{b, e\}$. Furthermore, matchfree$(A, B) = \{c\}$, and matchfree$(B, A) = \{f\}$.

\begin{figure}[htb]
\centering
\includegraphics[scale=.3]{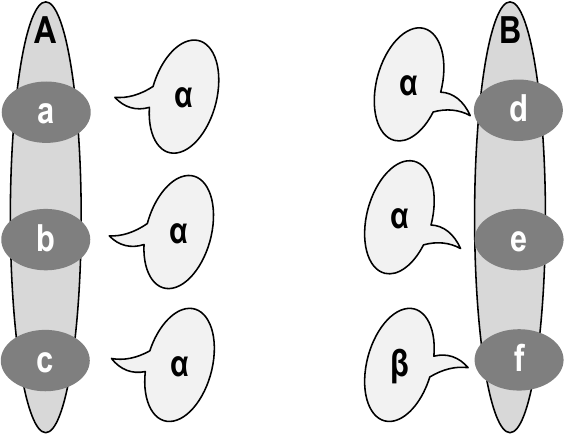}
\caption{Two interfaces, $A$ and $B$}
\label{fig:01}
\end{figure}

\subsection{Graphs with interfaces}\label{sec:2:3}

Based on the postulates in Section~\ref{sec:2:1}, it suggests itself to abstract a module as a \emph{graph}. This leaves the interpretation of nodes and edges open for many different interpretations, e.g., a typical such interpretation are Petri nets. 

(Directed) graphs are defined as usual: 

\begin{definition}[graph]
A \emph{graph} is a tuple $G = (V, E)$ consisting of a set $V$ of vertices and a set $E \subseteq V \times V$ of edges. 
\end{definition}

All forthcoming notions based on graphs apply to directed and undirected graphs alike. Mostly we assume graphs to be directed. But sometimes it is convenient to assume the undirected version. 

Any subset of nodes of a graph can serve as the carrier set of an interface. Two graphs $G$ and $H$ can be composed along an interface of $G$ and an interface of $H$, resulting again in a graph:

\begin{definition}[composition along interfaces] Let $G$ and $H$ be two graphs, let $A \subseteq G$ and $B \subseteq H$ be interfaces of $G$ and $H$. Then the \emph{composition of $G$ and $H$ along $A$ and $B$} is the graph $K$ where:
\begin{enumerate}

\item The nodes of $K$ are $(G \setminus A) \cup {(H \setminus B)} \; {\cup}$ matches${(A, B)} \; \cup$ matchfree${(A, B)} \; \cup$ matchfree$(B, A)$.

\item For each edge (x, y) of $G$ or of $H$,

\begin{itemize}

\item if $x$ and $y$ are both match free, then $(x, y)$ is an edge of $K$;
\item if $x$ is match free and ${y, y'}$ is a match, then $(x, \{y, y'\})$ is an edge of $K$;
\item if $\{x, x'\}$ is a match and $y$ is match free, then $(\{x, x'\}, y)$ is an edge of $K$;
\item if $\{x, x'\}$ and $\{y, y'\}$ are matches, then $\{\{x, x'\}, \{y, y'\}\}$ is an edge of $K$.
\end{itemize}

\end{enumerate}
\end{definition}

It suggests itself to conceive the matches of $G$ and $H$ as interior nodes of $K$ and the match free nodes as interface elements.

Figure~\ref{fig:02} shows two graphs with interfaces, and their composition along those interfaces. In fact, conceiving \emph{components} as graphs, and defining their composition as merging nodes of interfaces, appears an attractive choice, because it is quite general, and technically simple. Unfortunately, this idea has a serious drawback: this composition operator is not associative, i.e. for three graphs $G, H$, and $K$, the graphs $(G \compose H) \compose K$ and $G \compose (H \compose K)$ are in general not identical. Figure~\ref{fig:03} sketches an example; gates are drawn as lines, all with label $\alpha$. To overcome this problem, we suggest a slight variant of graphs with an interface, called \emph{modules}, as defined next.

\begin{figure}[htb]
\centering
\includegraphics[width=1\textwidth]{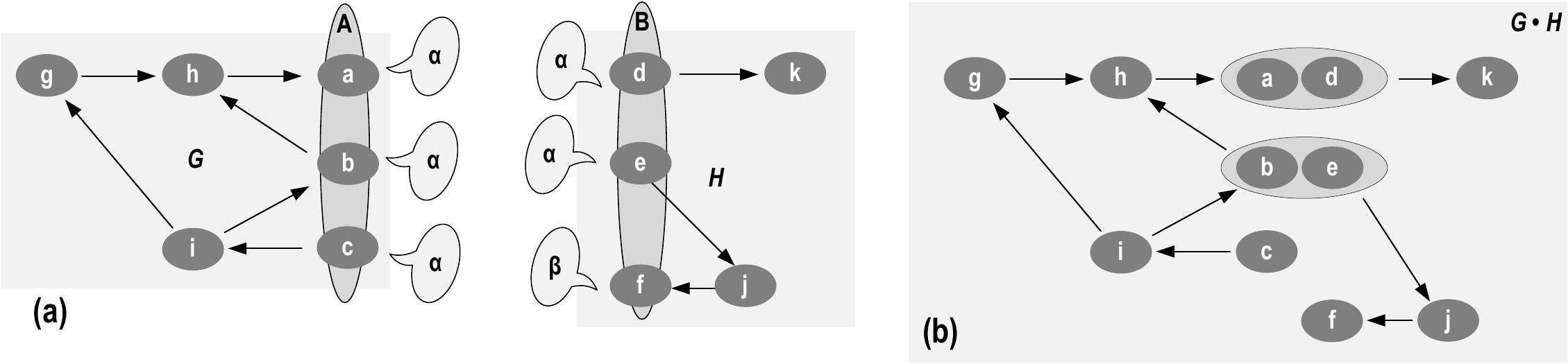}
\caption{Two graphs $G$ and $H$ and their composition $G \compose H$. (a) Graph $G$ with interface $A$ and graph $H$ with interface $B$. (b) Composition of $G$ and $H$ along the interfaces $A$ and $B$.}
\label{fig:02}
\end{figure}

\begin{figure}[htb]
\centering
\includegraphics[scale=.3]{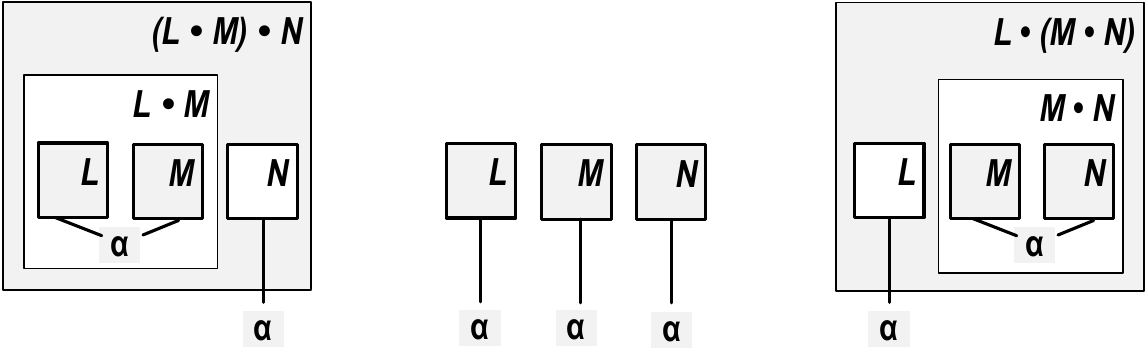}
\caption{$(L \compose M) \compose N$ and $L \compose {(M \compose N)}$ differ}
\label{fig:03}
\end{figure}

\subsection{Modules}\label{sec:2:4}
The decisive new idea here is the conception of a module as a graph with \emph{two} interfaces. Assuming a set $\Lambda$ of labels, as discussed in Section~\ref{sec:1:1}, we define:

\begin{definition}[module] A \emph{module $M$ over $\Lambda$} is a graph, together with two interfaces over $\Lambda$ which are subsets of nodes of $M$. The interfaces, written $^\ast M$ and $M ^\ast$, are the \emph{left} and \emph{right} interface of $M$. The nodes $M \setminus (^\ast M \cup M ^\ast)$ are the \emph{interior} of $M$.
\end{definition}

Graphically, a module is depicted as a box, surrounding the module’s graph. The left and the right interfaces are drawn onto the left and right margin of the box. Figure~\ref{fig:04}(a) shows examples. The above definition does not exclude a gate $g$ of a module to belong both to its left as well as to its right interface:

\begin{definition}[shared gate] For a module $M$, a \emph{gate} $g \in {^\ast M} \cap {M ^\ast}$ is \emph{shared} in $M$.
\end{definition}

In practical contexts, the feature of shared gates is occasionally useful.

\subsection{Composition of modules}\label{sec:2:5}
A fundamental feature of modules is their \emph{composition}. We first consider the case of modules without shared gates. To compose two such modules $M$ and $N$, their underlying graphs are composed along the interfaces $M ^\ast$ and $^\ast N$. In particular, the matches of $M ^\ast$ and $^\ast N$ go to the interior of $M \compose N$. The left interface $^\ast (M \compose N)$ of the composed module $M \compose N$ expands the left interface $^\ast M$ of $M$ by the match free elements of $^\ast N$. Likewise, the right interface $(M \compose N) ^\ast$ of the composed module $M \compose N$ expands the right interface $N ^\ast$ of N by the match free elements of $M ^\ast$.

\begin{definition}[composition of modules without shared gates] Let $M$ and $N$ be two modules over $\Lambda$ both without shared gates. Their \emph{composition} is the module $M \compose N$ over $\Lambda$ where:

\begin{itemize}

\item the graph of $M \compose N$ is the composition of the graphs of $M$ and $N$ along the interfaces $M ^\ast$ and $^\ast N$;

\item the left interface $^\ast (M \compose N)$ of $M \compose N$ is ${^\ast M} \; \cup$ matchfree$(^\ast N, M ^\ast)$. The gates of $^\ast M$ are ordered before the gates of matchfree$(^\ast N, M ^\ast)$;

\item the right interface $(M \compose N) ^\ast$ of $M \compose N$ is $N ^\ast \; \cup$ matchfree$(M ^\ast, ^\ast N)$. The gates of $N ^\ast$ are ordered before the gates of matchfree$(M ^\ast, ^\ast N)$.

\end{itemize}
\end{definition}

Figure~\ref{fig:02}(b) shows an example of composing two modules.

\begin{figure}[htb]
\centering
\includegraphics[width=1\textwidth]{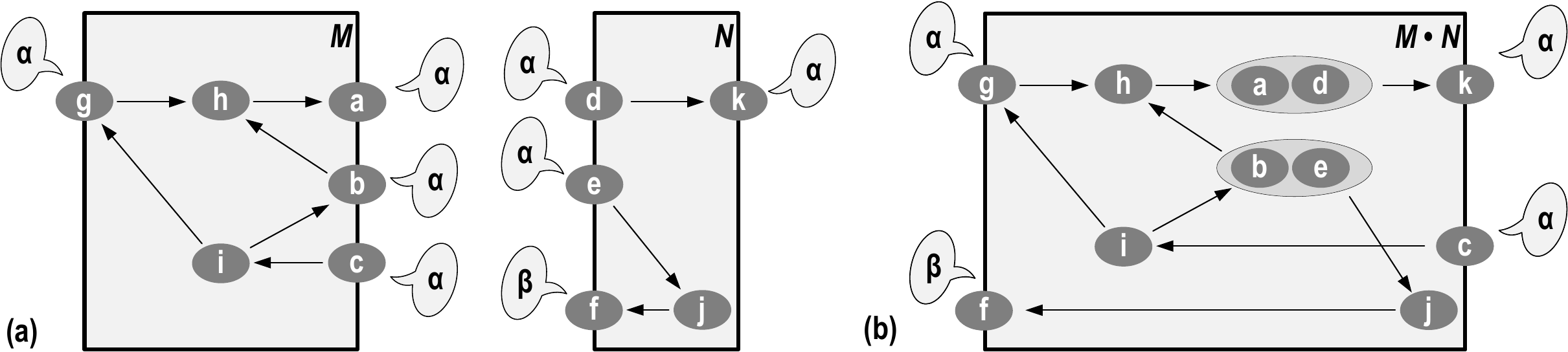}
\caption{(a) Modules $M$ and $N$, (b) their composition $M \compose N$}
\label{fig:04}
\end{figure}

Turning back to shared gates, it is obviously awkward to draw a shared gate $g$ of a module $M$ onto the left as well the right margin of the box of $M$. Instead, it turns out useful to draw $g$ twice, onto the left and the right margin of the box, and to link these drawings by a double line (reminding the equality symbol “=”). Figure~\ref{fig:05}(a) shows examples.

\begin{figure}[htb]
\centering
\includegraphics[width=1\textwidth]{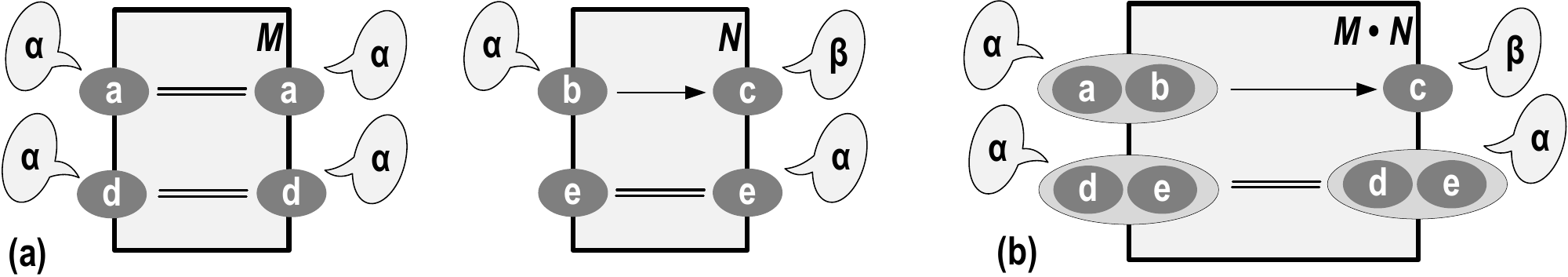}
\caption{Modules $M$ and $N$ with shared gates and their composition $M \compose N$. (a) Gates $a$ and $d$ belong to  $^\ast M$ and $M ^\ast$. Gate $e$ belongs to $^\ast N$ and $N ^\ast$. (b) The match $\{d, e\}$ belongs to $^\ast (M \compose N)$ and to $(M \compose N)^\ast$.}
\label{fig:05}
\end{figure}

Composition of modules with shared gates moves matches to interfaces:

\begin{definition}[composition of modules with shared gates] Let $M$ and $N$ be two modules over $\Lambda$. To define the \emph{composed module} $M \compose N$, the above composition of modules without shared gates is complemented as follows:

\begin{itemize}
\item each $x \in {^\ast M} {\cap} {M ^\ast}$ with a match $\{x,y\}$ of $M ^\ast$ and $^\ast N$ is in $^\ast (M \compose N)$ replaced by $\{x,y\}$;
\item each $y \in {^\ast N} \cap N ^\ast$ with a match $\{x,y\}$ of $M ^\ast$ and $^\ast N$ is in $(M \compose N) ^\ast$ replaced by $\{x, y\}$.
\end{itemize}

\end{definition}

\subsection{Equivalence}\label{sec:2:6}
The order of differently labeled gates is irrelevant for the composition of modules. This observation implies an equivalence on interfaces and on modules:

\begin{definition}[equivalent interfaces] Two interfaces $A$ and $B$ are \emph{equivalent} with $\phi$ if and only if 
$\phi: A \rightarrow B$ is a bijective mapping such that for all elements $a, b \in A$ holds:

\begin{enumerate}
\item the labels of a and of $\phi(a)$ coincide; 
\item if $a$ and $b$ are equally labeled, then $a < b$ in $A$, if and only if $\phi(a) < \phi(b)$ in $B$. 
\end{enumerate}
\end{definition}

As a shorthand, $B$ is written $\phi(A)$. For example, in Figure~\ref{fig:06}, the interfaces $A$ and $\phi(A)$ are equivalent.

\begin{figure}[htb]
\centering
\includegraphics[scale=.3]{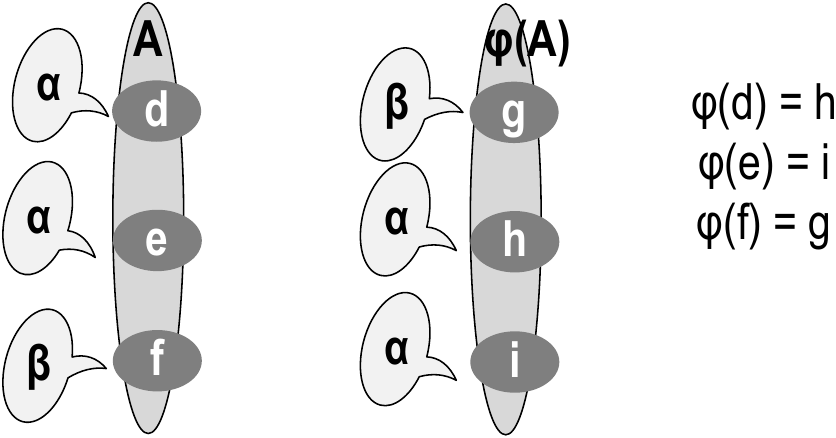}
\caption{Two equivalent interfaces A and $\phi(A)$}
\label{fig:06}
\end{figure}

Generally, it holds:

\begin{lemma}Let two interfaces $A$ and $B$ be equivalent with $\phi: A \rightarrow B$. Then for each interface $C$ holds:

\begin{enumerate}
\item $(a, c)$ is a match of $A$ and $C$, if and only if $(\phi(a), c)$ is a match of $B$ and $C$. 
\item $a \in A$ is match free with respect to $c \in C$, if and only if $\phi(a)$ is match free with respect to $c$. 
\end{enumerate}
\end{lemma}
 
Equivalence of interfaces canonically transfers to modules:

\begin{definition}[interface equivalent modules] Two modules $M$ and $N$ are \emph{interface equivalent}, if and only if the interfaces $^\ast M$ and $^\ast N$ are equivalent, and $M ^\ast$ and $N ^\ast$ are equivalent.
\end{definition}

\begin{theorem}
Let $M$ and $N$ be interface equivalent modules. Then for each module $L$ holds: $M \compose L$ and $N \compose L$ are interface equivalent, and also $L \compose M$ and $L \compose N$ are interface equivalent.
\end{theorem}

In fact, this Theorem shows that the relevant properties of modules are valid up to equivalence.

\subsection{Perfect Matches}\label{sec:2:7}
A pair $(M, N)$ of modules \emph{matches perfectly} if and only if $M$ and $N$ compose “without remainders”; i.e. if each gate of $M ^\ast$ finds a match in $^\ast N$ and, vice versa, each gate in $^\ast N$ finds a match in $M ^\ast$. This is the case if $M ^\ast$ and $^\ast N$ are equivalent: 

\begin{definition}[perfect match] Two modules $M$ and $N$ match perfectly, $(M, N)$ are a \emph{perfect match}, if and only if the interfaces $M^\ast$ and $^\ast N$ are equivalent.
\end{definition}

Composition of perfectly matching modules $M$ and $N$ “swallows” $M ^\ast$ and $^\ast N$ and transfers $^\ast M$ and $N ^\ast$ to the composed module:

\begin{theorem} Let $M$ and $N$ match perfectly. Then
\begin{enumerate}
\item Each gate of $M ^\ast$ and of $^\ast N$ belongs to a match of $M ^\ast$ and $^\ast N$.
\item matchfree$(M ^\ast, {^\ast N}) =$ matchfree$(^\ast N, M ^\ast) = \emptyset$.
\item $^\ast (M \compose N)$ and $^\ast M$ are equivalent; and also $(M \compose N) ^\ast$ and $N ^\ast$ are equivalent. 
\item If $^\ast M$ and $M ^\ast$ as well as $^\ast N$ and $N ^\ast$ share no gates then $^\ast (M \compose N) = {^\ast N}$, and $(M \compose N) ^\ast = N ^\ast$. 
\end{enumerate}
\end{theorem}

\section{Algebraic properties of module composition}\label{sec:3}
Composition of modules exhibits a bunch of important properties that hold for all kind of modules. The most important, fundamental property is the \emph{associativity} of the composition operator. Furthermore, composition “preserves” the involved modules; i.e., composition is \emph{cancellative}. Other properties hold only in special cases; in particular \emph{commutativity}, $M \compose N = N \compose M$, but also \emph{equidivisibility}. All these properties can be exploited at the design time of composed modules.

\subsection{Associativity}\label{sec:3:1}
Postulate 4 in Section~\ref{sec:2:1} demands that the composition of modules $M_1, \dots, M_n$ may just be written as $M_1 \compose M_2 \compose \dots \compose M_n$ without brackets. This requires \emph{associativity} of the composition operator, viz. the following Theorem:

\begin{theorem}Let $L, M$, and $N$ be three modules. Then $(L \compose M) \compose N = L \compose (M \compose N)$.
\end{theorem}

\cite{Reisig_19} presents a detailed proof of this Theorem. The example in Figure~\ref{fig:07} shows that the intermediate modules $L \compose M$ and $M \compose N$ look quite different. Associativity is in fact a fundamental property of modules. Without this property, the composition calculus would deserve few acceptance.

\begin{figure}[htb]
\centering
\includegraphics[scale=.3]{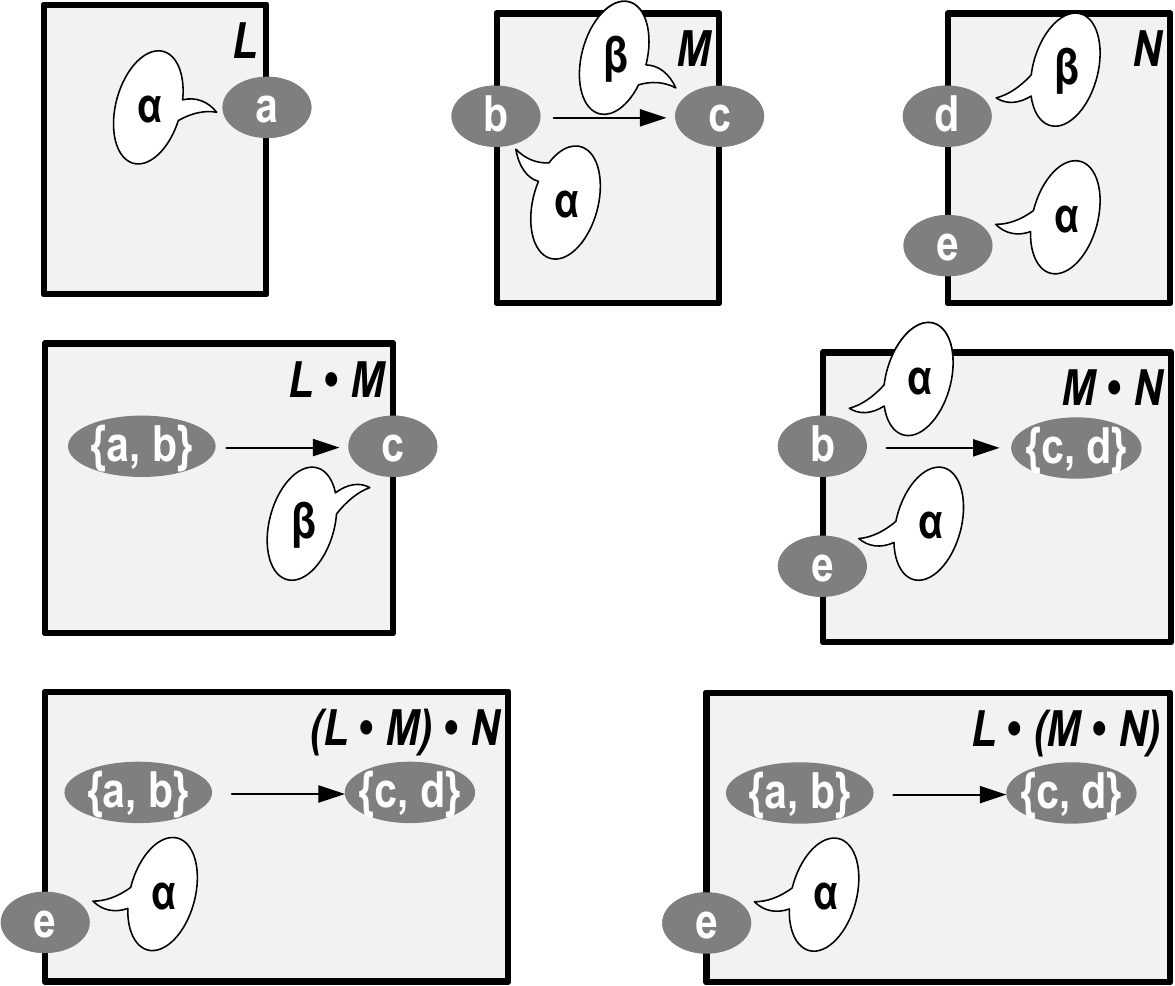}
\caption{Exemplifying the associativity theorem}
\label{fig:07}
\end{figure}

\subsection{Cancelativity}\label{sec:3:2}
Intuitively formulated, composition of modules does not resemble scrambled eggs, but composed words over an alphabet. Composition “preserves” the involved modules; in algebraic terms, composition is \emph{cancellative}:

\begin{theorem}
Let $L$, $M$, and $N$ be modules. 
\begin{enumerate}
\item If $L \compose M = L \compose N$, then $M = N$.
\item If $M \compose L = N \compose L$, then $M = N$.
\end{enumerate}
\end{theorem}

\begin{proof}
By contradiction, assume $M \neq N$. Then with no loss of generality $M$ has a node $m$ or an arc $(m, m')$ such that (i) $m$ is no node of $N$, or (ii) $(m, m')$ is no arc of $N$. In case of (i), $m$ is a node of $L \compose M$, but no node of $L \compose N$. In case of (ii), $(m, m')$ is an arc of $L \compose M$ or is embedded in a matching of $L \compose M$ but $(m, m')$ is no arc of $L \compose N$ or is not embedded in a matching of $L \compose N$. Hence, $L \compose M \neq L \compose N$.
\end{proof}

Cancelativity does not imply that the components $M$ and $N$ can be deduced from the composed module $M \compose N$. In fact, $M \compose N$ may be equal to $M' \compose N'$, with modules $M'$ and $N'$ that differ from $M$ and $N$. Nevertheless, knowing the \emph{cut} between $M$ and $N$, both modules can be re-computed from $M \compose N$. Details follow in Section~\ref{sec:3:3}.

\subsection{Commutativity}\label{sec:3:3}
Composition of modules is \emph{not commutative}, i.e., in general, $M \compose N \neq N \compose M$. However, $M \compose N$ and $N \compose M$ are interface equivalent in case the interfaces of $M$ and of $N$ share no labels. Figure~\ref{fig:08} shows four interface equivalent modules of commutative composition of two modules $M$ and $N$.

\begin{figure}[htb]
\centering
\includegraphics[scale=.3]{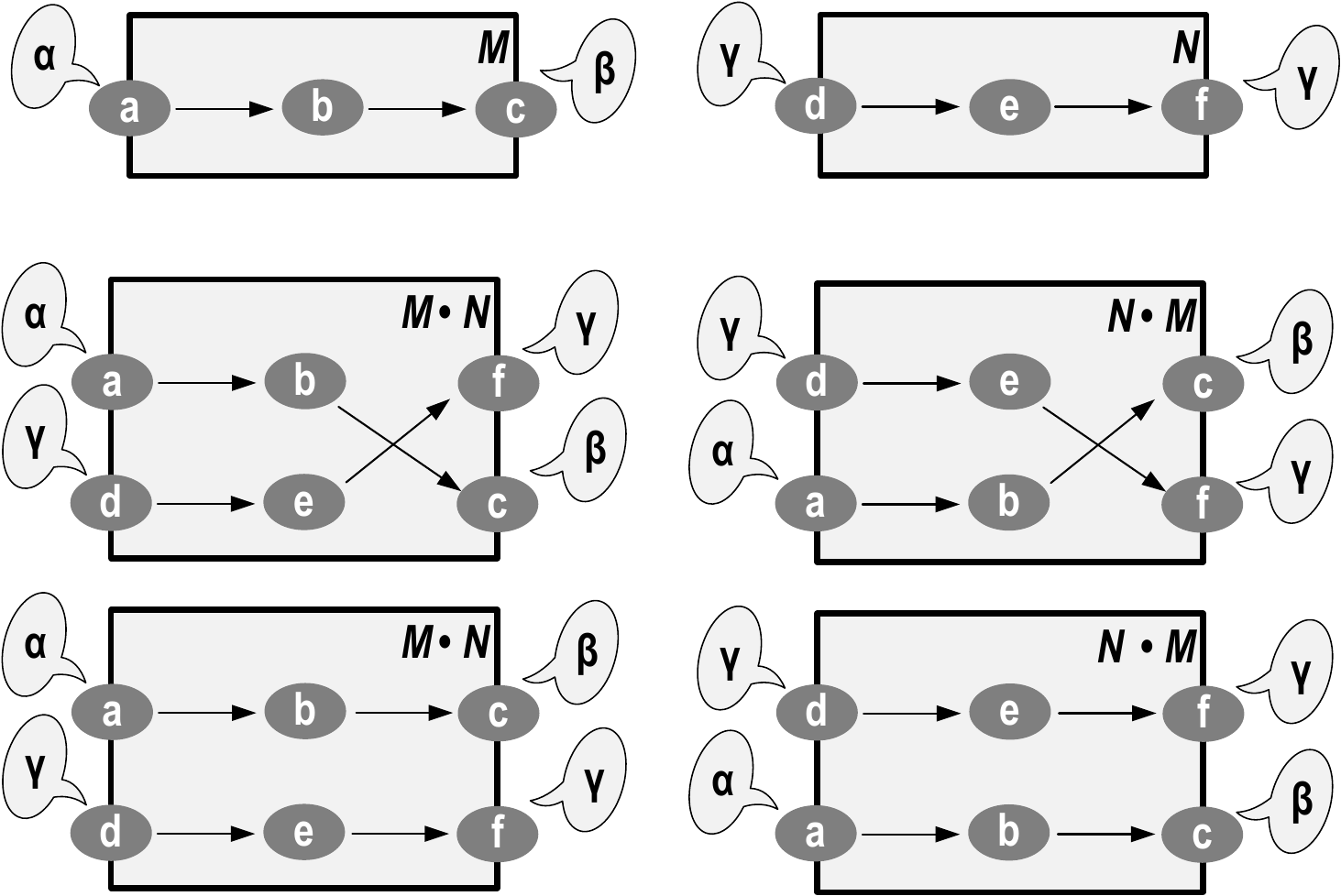}
\caption{Four interfaces equivalent representations of commutative composition of $M$ and $N$}
\label{fig:08}
\end{figure}

\begin{definition}
Two modules $M$ and $N$ are \emph{entangled} if and only if there are gates $m \in M ^\ast \cup {^\ast M}$ and $n \in N ^\ast \cup {^\ast N}$ which are equally labeled.
\end{definition}

\begin{theorem}For two modules $M$ and $N$ holds: $M \compose N$ and $N \compose M$ are interface equivalent if and only if M and N are not entangled.
\end{theorem}

\begin{proof}
To prove “$\rightarrow$”, by contradiction assume $M$ and $N$ are entangled. Case~1: there are gates $m \in M ^\ast$ and $n \in {^\ast N}$ which are equally labeled. Then $(m, n)$ is a match of $M ^\ast$ and $^\ast N$. Then $(m, n)$ goes to the interior of $M \compose N$, but not to the interior of $M \compose N$. Then $M \compose N$ and $N \compose M$ are not interface equivalent. Case~2: there are gates $m \in M ^\ast$ and $n \in N ^\ast$ which are equally labeled. Then in $M \compose N$, $m$ either has a matching partner $p$ in $^\ast N$, or $m$ is ordered after $n$ in $(M \compose N) ^\ast$. Then, in $N \compose M$, $p$ is no matching partner in $^\ast N$ or $m$ is ordered before $n$ in $(N \compose M) ^\ast$. Then $M \compose N$ and $N \compose M$ are not interface equivalent. The cases of equally labeled gates $m \in M ^\ast$ and $n \in {^\ast N}$ and of $m \in {^\ast M}$ and $n \in {^\ast N}$ are treated analogously. 

To prove “$\leftarrow$”, assume $M$ and $N$ are not entangled. In this case, $^\ast (M \compose N) = {^\ast M} \cup {^\ast N} = {^\ast N} \cup {^\ast M} = {^\ast (N \compose M)}$. Likewise, $(M \compose N) ^\ast = (N \compose M) ^\ast$.
\end{proof}

Commutativity is relevant in many contexts. Such contexts can be characterized in terms of \emph{patterns}. A typical pattern includes a module $M$ that may occur in $n$ parameterized versions, $M_1, \dots, M_n$ with pairwise disjoint parameters. The modules $M_i$ can safely be composed in any order. Another pattern exploits hierarchical concepts: In Figure~\ref{fig:09}, composition of the modules $L$ and $N$ is not commutative nor is the composition of $M$ and $N$ commutative. However, $N$ is commutative with $P \defeq L \compose M$.

\begin{figure}[htb]
\centering
\includegraphics[width=1\textwidth]{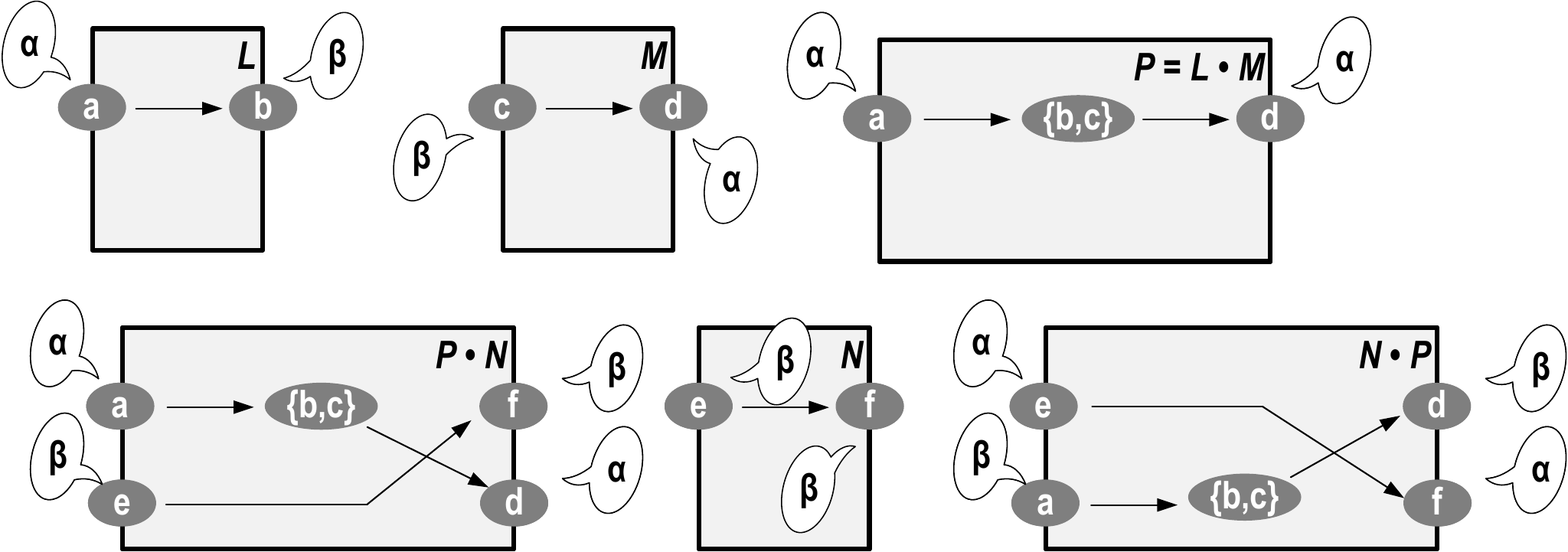}
\caption{Hierarchy of modules}
\label{fig:09}
\end{figure}

\subsection{Equidivisibility}\label{sec:3:4}
As discussed in Section~\ref{sec:3:3}, two modules $M$ and $N$ commute if the label sets of their interfaces are disjoint. Now we consider an almost complementary case, where label sets are not disjoint.

\begin{definition}[precedence] Let $M$ and $N$ be modules, where $M ^\ast$ and $^\ast N$ share at least on label. Then $M$ \emph{precedes} $N$.
\end{definition}
	
Now assume two pairs of preceding modules, such that the composition of each pair results in the same module. Then the components share an overlapping module: 

\begin{theorem}
Let $K, L, M, N$ be modules, let $K$ precede $L$ and $M$ precede $N$. Furthermore, assume $K \compose L = M \compose N$. Then there exists a module $P$ such that either $K \compose P = M$ and $P \compose N = L$, or $M \compose P = K$ and $P \compose L = N$.
\end{theorem}

\begin{proof}
This Theorem is an application of Levi's Lemma, as formulated e.g. in \cite{de_Luca_99}.
\end{proof}

This theorem structures the refinements of modules: Assume a module $Q$ is refined in two different ways, $Q = K \compose L$ and $Q = M \compose N$. Then there exists a unique most abstract refinement of $Q$ in terms of $K, L, M$, and $N$. To construct this refinement, assume with no loss of generality that $K$ is a prefix of $M$. Then there is a module $P$ with $K \compose P = M$. Then $K \compose P \compose N = Q$. Figure~\ref{fig:10} sketches this construction.

\begin{figure}[htb]
\centering
\includegraphics[scale=.3]{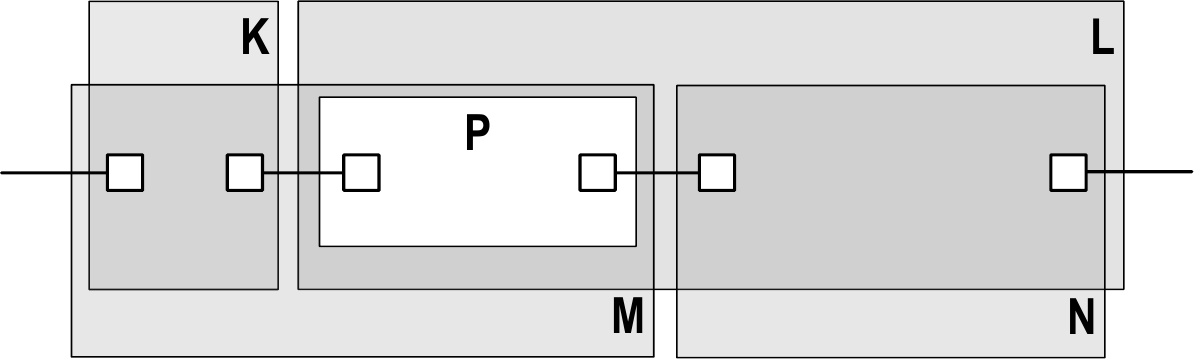}
\caption{Let $Q \defeq K \compose L = M \compose N$. $M$ and $L$ overlap with module $P$. Furthermore, $Q = K \compose P \compose N$.}
\label{fig:10}
\end{figure}

\subsection{Neutral module}\label{sec:3:5}
There exists a \emph{neutral} module, i.e. a module that has no effect upon composition:

\begin{definition}[neutral module] The module $E \defeq (\emptyset, \emptyset)$ is denoted as the \emph{neutral module}. 
\end{definition}

\begin{lemma}
For all modules $M$ holds: $M \compose E = E \compose M = M$. 
\end{lemma}

\section{Classes of modules}\label{sec:4}
This Section focuses on classes of modules with distinguished properties. We start with the observation that some, but not all, interesting sets of modules can be generated from a finite set of “basic” modules. Secondly, details of the interior of modules are not always important; sometimes only the interfaces count. The interior is abstracted away and replaced by one or no interior vertex. We finish with a particular class of modules that mimic the classical set of words over a given alphabet.

\subsection{Finitely generated sets of modules}\label{sec:4:1}
Each given finite set $T$ of modules generates a set $S$ of modules: starting with $S \defeq T$, iteratively, for each two modules $M$ and $N$ in $S$, the composition $M \compose N$ is added to $S$. Formulated differently: 

\begin{definition}
A set $S$ of modules is \emph{generated by a set $T$ of modules} if $S$ is the smallest set of modules that contains $T$ and is closed under composition.
\end{definition}

\begin{definition}
A set $S$ of modules is \emph{finitely generated} if $S$ is generated by a finite set $T$ of modules.
\end{definition}

\begin{theorem}
The set $\underline{M}$ of all modules over $\Lambda$, as introduced in Section~\ref{sec:2:4}, is not finitely generated. 
\end{theorem}

\begin{proof}
To prove this observation, by contradiction assume a finite set $T$ of modules that generates $\underline{M}$. Assume each module in $T$ has at most $m$ vertices and $n$ edges. Then, according to the definition of the composition operator, for each $k \in \mathbb{N}$, the $k$-fold composition $K_1 \dots K_k$ of Modules $K_i$ in $T$ has at most $k_m$ vertices and $k_n$ edges. But there are many graphs with $k_n$ vertices and more than $k_n$ edges.
\end{proof}

This Theorem implies the intuitive insight that the number of edges in a “big” system, composed of “many” modules, is sparse.

\subsection{Atomic modules}\label{sec:4:2}
In many contexts, only the architecture of a system matters: The interior of each module is hidden.

\begin{definition}
A module $M$ is \emph{atomic}, $M$ is an \emph{atom}, if and only if $M$ has exactly one interior vertex, $p$, linked to each gate $g$ of ${^\ast M}$ or $M ^\ast$ by an undirected edge. The vertex $p$ is the \emph{name of $M$}.
\end{definition}

A module $M$ can be abstracted to its corresponding \emph{atom}, $\langle M \rangle$:

\begin{definition}
For a given module $M$, let $\langle M \rangle$ denote the atom where $^\ast \langle M \rangle = {^\ast M}$ and $\langle M \rangle ^\ast = M ^\ast$. The \emph{atomic module $\langle M \rangle$} is the atom of $M$. 
\end{definition}

Composition of atomic modules is never atomic. Composition $\langle M_1 \rangle \compose \langle M_2 \rangle \compose$ $\dots$ $\compose \langle M_n \rangle$ of atoms $\langle M_i \rangle$ of modules $M_i$ can store the history of composition $M_1 \compose M_2 \compose \dots \compose M_n$.

Abstraction and composition can be applied in either order:

\begin{lemma}
Let $M$ and $N$ be modules.
\begin{enumerate}

\item $\langle\langle M \rangle\rangle = \langle M \rangle$.

\item $\langle\langle M \rangle \compose \langle N \rangle\rangle = \langle M \compose N\rangle$.

\item $^\ast \langle M \compose N\rangle = {^\ast (\langle M \rangle \compose N)} = {^\ast (M \compose \langle N \rangle)} = {^\ast (\langle M \rangle \compose \langle N \rangle)}$.

\item $\langle M \compose N\rangle ^\ast = (\langle M\rangle \compose N) ^\ast = (M \compose \langle N \rangle) ^\ast = (\langle M\rangle \compose \langle N\rangle) ^\ast$. 
\end{enumerate}
\end{lemma}

\subsection{Abstract modules}\label{sec:4:3}
In many contexts, the interior of modules is irrelevant and is abstracted away entirely. Only the interfaces remain. 

\begin{definition}
A module $M$ is \emph{abstract}, if and only if it has no interior vertices and no edges. 
\end{definition}

Obviously, an abstract module is fully specified by its both interfaces. A module $M$ yields its corresponding \emph{abstraction}, $[M]$:

\begin{definition}
For a given module $M$, let $[M]$ denote the \emph{empty module} with $^\ast [M] = {^\ast M}$ and $[M] ^\ast = M ^\ast$. $[M]$ is the \emph{abstraction of $M$}.
\end{definition}

Composition retains abstraction; i.e. modules can be abstracted and composed in either order:

\begin{lemma}
Let $M$ and $N$ be modules.
\begin{enumerate}
\item $[[M]] = [M]$.
\item $^\ast [M \compose N] = {^\ast ([M] \compose N)} = {^\ast (M \compose [N])} = {^\ast ([M] \compose [N])}$.
\item $[M \compose N] ^\ast = ([M] \compose N) ^\ast = (M \compose [N]) ^\ast = ([M] \compose [N]) ^\ast $.
\end{enumerate}
\end{lemma}

Composition of abstractions is again abstract only in special cases:

\begin{lemma}
Let $M$ and $N$ be interface modules. Then $M \compose N$ is an interface module, if and only if $M ^\ast$ and $^\ast N$ have no match.
\end{lemma}

\subsection{Word modules}\label{sec:4:4}
As a very special case, for the alphabet $\Sigma = \{a, b, \dots, z\}$, let $M_a, M_b, \dots, M_z$ be sequentially controlled and abstract modules as in Figure~\ref{fig:11}. They generate the free monoid of words of formal languages, as in informatics.

The set of word modules over an alphabet $\Sigma$ is finitely generated.

\begin{figure}[htb]
\centering
\includegraphics[width=1\textwidth]{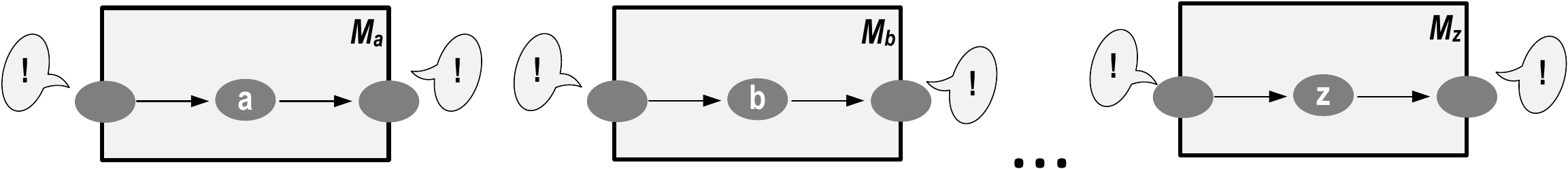}
\caption{Each symbol gets its module}
\label{fig:11}
\end{figure}

\section{Related work}\label{sec:5}

The quest for modules and their composition is of paramount importance for any kind of system models. Starting with the famous work of Parnas \cite{Parnas_72}, literature on modules and composition of modules is voluminous. In most contributions, models and programs are equipped with interfaces, and composition is defined along those interfaces. Typical examples include:

\begin{itemize}
\item \emph{Applied modeling frameworks}: The most influential modeling framework, UML \cite{Jacobson_99a}, suggests, on a semi-formal level, a particular kind of composition: a set of mutually unrelated modules are composed and the result constitutes a larger module. Altogether, this yields a tree-shaped architecture, with vertically unrelated branches. The composition calculus is much more flexible in this respect.

\item \emph{Web services}: There are plenty of approaches to \emph{web service} composition; composition can be agent-based, synchronous, asynchronous, or based on interaction protocols \cite{Petrie_16}. The composition calculus covers those composition variants by means of \emph{adapters}, as discussed in \cite{Fettke_Reisig_24}.

\item \emph{Interface languages}: An \emph{interface language} represents interface descriptions in a language-independent way, in order to link programs in different languages \cite{Lamb_87}. It would be interesting to re-formulate interfaces of those languages as domain specific instantiations of the composition calculus. 
\end{itemize}

More abstract calculi which support modules and their composition include \emph{Hoare’s communicating sequential processes} (CSP) \cite{Hoare_85}, and the many versions of process calculi [6]. These calculi come with associative composition, but at the price of non-determinism. In Section~\ref{sec:2:1} it was discussed already that process calculi are tightly related to the composition calculus, as in both cases, composition is technically defined by the merging of complementarily labeled interface elements. This also applies to variants and abstractions of process calculi such as Montanari's tiles \cite{Bruni_04} and many more approaches.

Composition of Petri nets has a long tradition, with many contributions. We just mention some relevant ideas of the last three decades: \cite{Christensen_Petrucci_93} composes many Petri net modules in one go, merging all equally labeled places as well as transitions. Associativity of composition is not addressed; but it seems obvious that this composition operator is not associative. A long-standing initiative with many variants is algebraic calculi for Petri nets, such as the \emph{box calculus} and \emph{Petri net algebras} in various forms \cite{Best_et_al_01}. In the spirit of process algebras, these calculi define classes of nets inductively along various composition operators, merging equally labeled transitions. Associativity of composition is mostly assumed, but rarely explicitly discussed. \cite{Baldan_08} defines composition $A \compose B$ of “composable” nets $A$ and $B$ in a categorical framework, with $A$ and $B$ sharing a common net fragment, $C$. The effect of composition on distributed runs is studied. The issue of associativity of composition is not addressed. In the composition calculus, one would formulate the net fragment $C$ as an adapter, constructing $A' \compose C \compose B'$. Here, $A'$ and $B'$ are $A$ and $B$ without $C$. \cite{Kindler_Petrucci_09} suggests a general framework for “modular” Petri nets, i.e. nets with features to compose a net with its environment. A module may occur in many instances. Associativity of composition is not discussed in this paper, but implicitly assumed. 

Petri nets with two-faced interfaces are discussed in \cite{Rathke_et_al_14,Sobocinski_10}. They suggest \emph{Petri nets with boundaries} (PNB), with two composition operators. A PNB resembles a module in that it has a left and a right interface. One of the composition operators recalls the case of transition interfaces with all transitions labeled alike, the other one equals the case of disjoint labels. Interestingly, both operators are associative, but in general not commutative. It might be worthwhile to formulate the paper’s results on compositional reachability in the framework of the composition calculus, formulating specific assumptions on the interior of modules.

The first ideas for modules and the composition operator occurred for runs of Petri nets in \cite{Reisig_05}. Associativity of the composition operator was proven in \cite{Reisig_19}. Double-faced interfaces for modeling techniques other than Petri nets can be found in the literature, albeit in versions that are more specialized than the composition calculus. A typical example is the “piping” or “chaining” operator $P >> Q$ for the language CSP \cite{Roscoe_10}. Associativity of $>>$ is assumed without discussion.

There are numerous versions of equipping classical automata with mechanisms for communication; e.g. \cite{de_Alfaro_Henzinger_01,Fendrich_Luttgen_19}. Usually, composition includes semantic requirements and properties, yielding lots of variants of compositions $A \compose B$. The composition calculus suggests to locate semantic requirements in an adapter, $C$, in a composition $A \compose C \compose B$. Composition then can entirely be defined in terms of the interfaces of $A$, $B$, and $C$.

\section{Conclusion}\label{sec:6}
The aim of this contribution is to fortify the claim that the composition calculus is not just another composition operator, but that this operator is a fundamental basis for any composition principle. The composition calculus turns out as a further example of a “good” algebraic theory: A minimum of assumptions yield a multitude of non-trivial properties, i.e. graphs, and labeled, ordered subsets of vertices. Conceptually, the composition calculus defines composition purely by means of interfaces; the involved modules’ interior remains more or less untouched. This allows to consider composition independent of specific instantiations. The insights given here then of course hold for all instantiations.

\subsubsection*{Acknowledgment}
We deeply appreciate the careful and thoughtful hints and comments of the referees.


\bibliographystyle{splncs04}
\bibliography{main}

%\newpage
%\subsection*{Appendix}

%The following two figures depict a distributed run of an unfulfilled and fulfilled service request of a customer.

%\begin{figure}[htb]
% \centering
% \includegraphics[width=0.6\textwidth]{figures/service_system_run_service_erfolglos.png}
% \caption{Distributed run of an unfulfilled service request}
% \label{fig:run_unfulfilled}
%\end{figure}

%\begin{figure}[htb]
% \centering
% \includegraphics[width=1\textwidth]{figures/service_system_run_erfolgreich.png}
% \caption{Distributed run of a fulfilled service request}
% \label{fig:run_fulfilled}
%\end{figure}

\end{document}